\documentclass[manuscript,screen]{acmart}
\AtBeginDocument{%
  \providecommand\BibTeX{{%
    \normalfont B\kern-0.5em{\scshape i\kern-0.25em b}\kern-0.8em\TeX}}}

\usepackage{subfigure}
\usepackage{graphicx}
\usepackage{textcomp}
\usepackage{amsthm,amsmath}
\usepackage{mathrsfs}
\usepackage{algorithm}
\usepackage{algorithmic}
\usepackage{caption}
\usepackage{booktabs}
\usepackage{threeparttable}
\usepackage{multirow}
\usepackage{blindtext}
\usepackage{xpatch}

\usepackage{xcolor}

\renewcommand{\algorithmicrequire}{\textbf{Input:}}

\makeatletter
\xpatchcmd{\ps@firstpagestyle}{Manuscript submitted to ACM}{}{\typeout{First patch succeeded}}{\typeout{first patch failed}}
\xpatchcmd{\ps@standardpagestyle}{Manuscript submitted to ACM}{}{\typeout{Second patch succeeded}}{\typeout{Second patch failed}}\@ACM@manuscriptfalse% Also in titlepage
\makeatother

\renewcommand\footnotetextcopyrightpermission[1]{} % removes footnote with conference information in first column
\setcopyright{none}
\begin{document}
% Main title of the paper
%\title [mode = title]{Local Hidden Community Detection through Iterative Reduction Local Spectra}                      
\title {Uncovering the Local Hidden Community Structure in Social Networks} 

\author{Meng Wang}
\authornote{The first two authors contributed equally.}
\email{mengwang233@hust.edu.cn}
% \orcid{1234-5678-9012}
\affiliation{%mengwang233@hust.edu.cn
  \institution{School of Computer Science \& Technology, Huazhong University of Science and Technology}
  \city{Wuhan}
  \state{Hubei}
  \country{China}
}

\author{Boyu Li}% afterslby@hust.edu.cn
\email{afterslby@hust.edu.cn}
\affiliation{%
  \institution{School of Computer Science \& Technology, Huazhong University of Science and Technology}
  \city{Wuhan}
  \state{Hubei}
  \country{China}
}

\author{Kun He}
\authornote{Corresponding author.}% Email: brooklet60@hust.edu.cn}
\email{brooklet60@hust.edu.cn}
\affiliation{%
  \institution{School of Computer Science \& Technology, Huazhong University of Science and Technology}
  \city{Wuhan}
  \state{Hubei}
  \country{China}
}

\author{John E. Hopcroft}
\email{jeh@cs.cornell.edu}
\affiliation{%
  \institution{Department of Computer Science, Cornell University}
  \city{Ithaca}
  \state{NY}
  \country{USA}
}

% \makeatletter
% \let\@authorsaddresses\@empty
% \makeatother

\renewcommand{\shortauthors}{Wang et al.}

\begin{abstract}
Hidden community is a useful concept proposed recently for social network analysis. Hidden communities indicate some weak communities whose most members also belong to other stronger dominant communities. Dominant communities could form a layer that partitions all the individuals of a network, and hidden communities could form other layer(s) underneath. These layers could be natural structures in the real-world networks like students grouped by major, minor, hometown, etc. 
To handle the rapid growth of network scale, in this work, we explore the detection of hidden communities from the local perspective,  and propose a new method that detects and boosts each layer iteratively on a subgraph sampled from the original network. We first expand the seed set from a single seed node based on our modified local spectral method and detect an initial dominant local community. Then we temporarily remove the members of this community as well as their connections to other nodes, and detect all the neighborhood communities in the remaining subgraph, including some ``broken communities''  that only contain a fraction of members in the original network. 
The local community and neighborhood communities form a dominant layer, and by reducing the edge weights inside these communities, we weaken this layer's structure to reveal the hidden layers. Eventually, we repeat the whole process and all communities containing the seed node can be detected and boosted iteratively. 
We theoretically show that our method can avoid some situations that a broken community and the local community are regarded as one community in the subgraph, leading to the inaccuracy on detection which can be caused by global hidden community detection methods. Extensive experiments show that our method could significantly outperform the state-of-the-art baselines designed for either global hidden community detection or multiple local community detection.    
\end{abstract}



\maketitle
\thispagestyle{empty}
\pagestyle{empty}

\section{Introduction}
\newcommand{\etal}{\textit{et al}. }
\newcommand{\ie}{\textit{i}.\textit{e}.}
\newcommand{\eg}{\textit{e}.\textit{g}.}
\newcommand{\etc}{\textit{etc}}
\newcommand{\st}{\textit{s}.\textit{t}. }

Social network analysis has received significant attention from the researchers of various fields for decades. As an important topic for social network analysis, the task of community detection aims to uncover clusters of nodes (individuals in social networks represented by graphs), termed communities, whose interior closeness is much higher than their exterior connections to the remaining nodes.
Uncovering the community structure in social networks not only contributes to revealing the most valuable interactions and behaviors of the individuals, but also helps to understand the overall pattern and features of the networks. 
Technologies for community detection can be classified into two categories, namely global community detection and local community detection. Global community detection finds communities covering all the nodes of the network~\cite{newman2004finding,rosvall2008maps,coscia2012demon, zhu2020community}, while local community detection aims to seek for cohesive group(s) containing a single query seed or a few seed members~\cite{clauset2005finding,andersen2006local,li2015uncovering, luo2020local}.

From the perspective of global community detection, other than the enormous works for disjoint community detection and overlapping community detection, there is another kind of community structure called the hidden community structure, which has attracted increasing interests in the field of social networks in recent years~\cite{young2015shadowing,he2015revealing, he2018hidden,gong2018finding}. 
A community is called the ``hidden community'' if most of its members also belong to other stronger communities as evaluated by some community scoring functions (\eg, modularity~\cite{newman2004finding}, conductance~\cite{shi2000normalized} or cut ratio~\cite{fortunato2010community}). 
Such weak communities are often overlooked by conventional community detection algorithms that tend to detect the dominant communities with relatively strong structures.
Classified by the strength, communities can be organized in ``layers'' where each layer contains a set of communities that partition all the nodes.
% , or cover all the nodes in the way that the communities slightly overlap with each other.  

%Conventional community detection algorithms mostly utilize the interactions inside the network to obtain only one community division. However, 
%In the real world, people usually have several different kinds of social relations, and they can be considered as multiple layers of a network, where each layer is a partition, corresponding to a set of communities~\cite{he2015revealing, he2018hidden}. In the meantime, the strength of these layers are not always the same. A community with weaker structure may be covered by other stronger communities, making this community ``hidden'' and it is much harder to detect these hidden communities. 

We observe that a natural layer of communities can be formed by a partition of the individuals grouped by their attributes in social networks. For instance, consider a real-world network Rice\_Ugrad with 1220 undergraduate students in the Rice University~\cite{mislove2010you}. All students in the network have several attributes and we assign a part of them to the same communities when they have the same value of an attribute and exhibit denser interior connections than the exterior connections. When students are grouped by colleges, we see a partition with nine communities, which yields a high modularity of 0.384. And there is another set of four sizable and three trivial communities when grouping students by matriculation years, whose modularity is 0.259. The two sets of communities form two layers of the network, and the corresponding re-permuted adjacency matrices are illustrated in Figure \ref{fig:Ugrad}. Although each layer has its own prominent structure, it may not contribute to each other's partition, and the connections inside the communities from one layer can even be considered as background noises from the perspective of the other layer. On the one hand, if we want to find students entering the university in the same year, the interference from connections in their colleges must be eliminated. On the other hand, when we try to detect the communities within colleges, though they are comparatively stronger, the effect caused by hidden matriculation-year structures cannot be ignored as well.

\begin{figure}[t]%htbp]
    \centering
        \subfigure[Grouped by college]{
        \begin{minipage}[t]{0.5\linewidth}
        \centering
        % \flushleft
        \includegraphics[width=4.5cm]{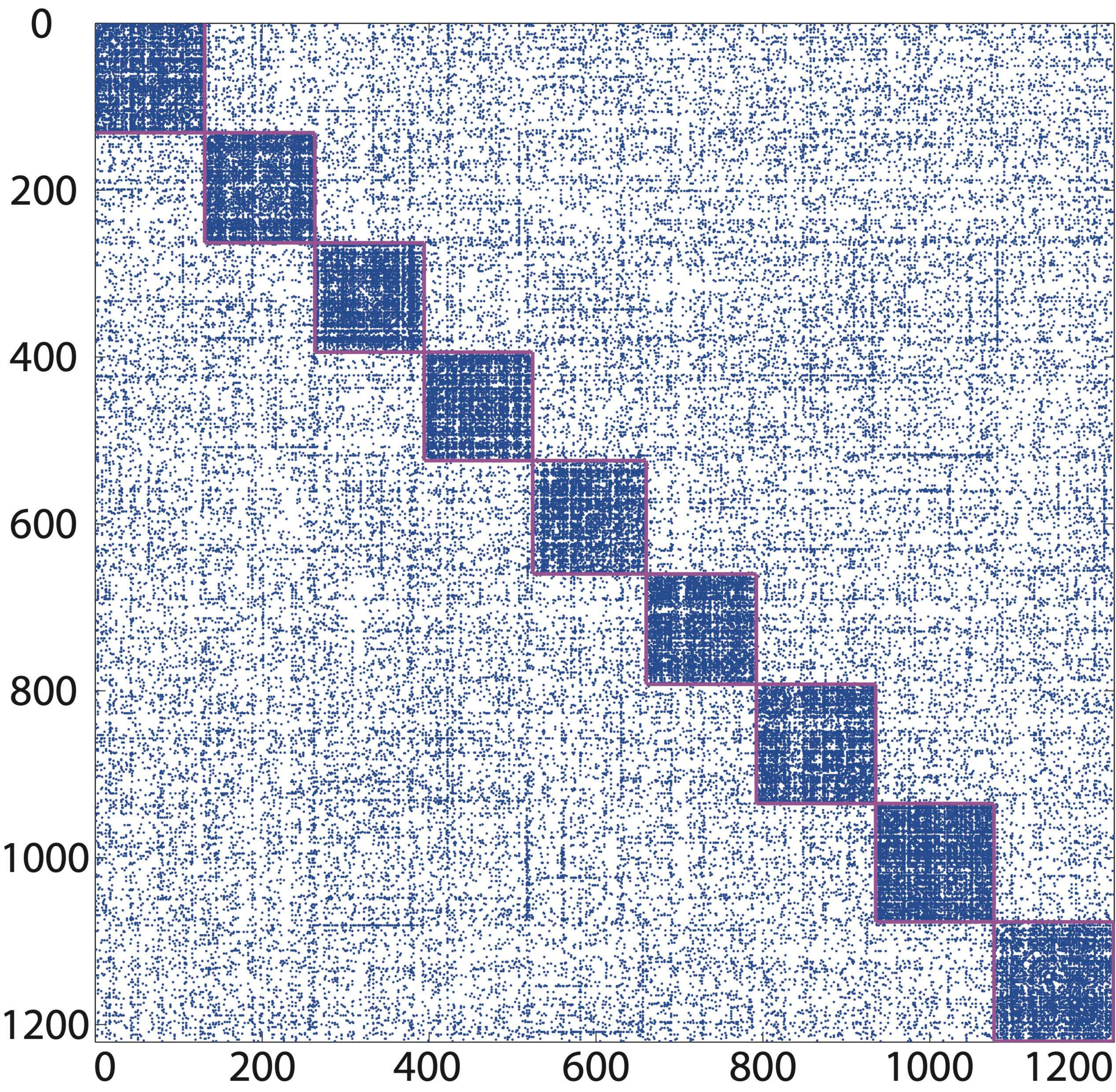}
        \end{minipage}
        }%
        \subfigure[Grouped by matriculation year]{
        \begin{minipage}[t]{0.5\linewidth}
        \centering
        % \flushleft
        \includegraphics[width=4.5cm]{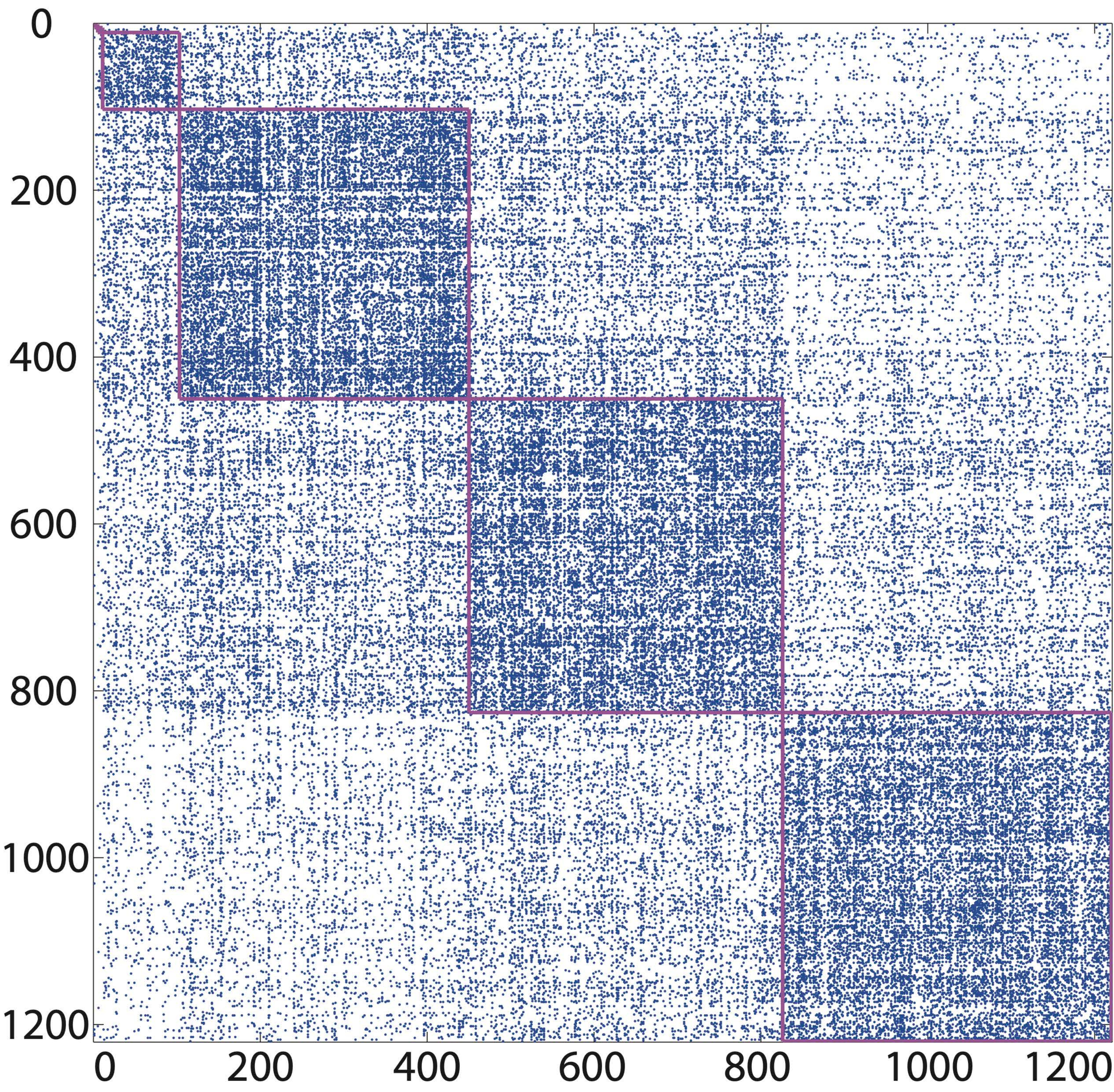}
        \end{minipage}
        }%
    \caption{Two layers of the Rice\_Ugrad network.}
    \label{fig:Ugrad}
\end{figure}

%\HK{You did not introduce the work of local community, and multiple local community detection. What is the difference of your work as compared with the multiple local community detection?}

Meanwhile, with the rapid growth of network scale, the global community detection becomes very costly or even unavailable in terms of the computability for huge networks, especially when all we want to obtain are just one or several members' local relations. This leads to a growing hot topic of the local community detection, for which the researchers focus on finding a local community for one single seed node or a few seed nodes. More recently, some researchers started to investigate the multiple local community detection of finding several overlapping local communities that the query node belongs to.
However, to the best of our knowledge, there is no research in detecting the local hidden communities. When addressing the local community detection problem in a network with hidden structure, we also aim to discover several communities containing the seed, but with different degrees of cohesiveness. Instead of having similar strengths and simply overlapping with each other, these communities are distributed into different layers with various strengths. And if we simply apply a multiple local community detection algorithm, the hidden communities covered by stronger ones are hardly to be uncovered.

% In a network with multiple layers of communities, a social member can be affected greatly by anyone around, from either a strong or a weak community.
Assuming an example that the police has captured a suspect, and they want to identify his or her partners in crime. Usually the criminal group has less often contacts than that among families or colleagues, but at the moment the relations within this group are most valuable. Under the circumstance if we apply a conventional local community detection method or even one designed for overlapping situations, we are more likely to keep our eyes on the closest people to the suspect but unable to dig out the exact criminal group due to its weak structure overwhelmed by other stronger structures. This reflects the necessity of addressing the hidden community detection problem.

The problem of detecting hidden structure has been noticed and discussed by a few works in recent years. He \etal~\cite{he2015revealing, he2018hidden} come up with the HICODE algorithm, which detects and weakens a layer to reveal the hidden structures. Gong \etal~\cite{gong2018finding} design an embedding method to implement the similar functionality. However, the methods proposed in these works are all designed for conducting the global community detection. The limitations of these methods, which we mentioned above, significantly affect their efficiency in large networks.
% In the age of big data, with the rapid growth of the network scale, it is less efficient to apply algorithms on the overall network, especially when all we want to find are just one or several persons' local relations. 
Then, what if we perform a sampling process in the original network and use a global hidden community method in the sampled subgraph?
Although the sampling process is very likely to reserve the target local communities but it would segment a portion of other communities to have incomplete structures. 
%and the members whose former partners are abandoned could be assigned to other groups, 
%which also affect the accuracy of the detection. 
Such imbalance would weaken the detection quality of the methods in sampled subgraphs.
Consequently, a local community detection method particularly designed for networks with multiple layers of community structures is absent.

In this work, we propose an effective algorithm called the Iterative Reduction Local Spectral (IRLS) method to uncover the multiple local communities of a query node in different layers, including dominant and hidden local communities.
To the best of our knowledge, our paper is the first work to address the hidden communities from the local perspective. 
For the single-layer local community detection, we add an adaptive augmentation method for the seed set to a local spectral algorithm, and set the truncating rule by the weighted local modularity when the actual size of the target community is unknown. Then the modified local spectral algorithm is used in the sampled subgraph to uncover the group containing the seed in the dominant layer. The members of the detected local community and all connections related to them are removed from the network temporarily for the detection of the remaining communities of the current layer. 
Then we reduce the edge weights inside the discovered communities of the layer to decrease their cohesiveness to be the same with the surroundings. In this way, we break the dominant layer's structure and the hidden layer emerges, which offers the chance to discover another type of relations of the seed node. Such steps are performed iteratively until the local communities are detected. Experiments on both synthetic and real-world datasets verify the effectiveness of our method.

Our main contributions of this work include:
\begin{itemize}
    \item We introduce a new problem called the local hidden community detection, whose purpose is to uncover the dominant local community as well as the hidden local communities for a single query seed on networks with multiple layers of community structures. 
    \item We propose an Iterative Reduction Local Spectral (IRLS) algorithm for this problem, in which we seek for a partition of the remaining network after the local community is uncovered, and then weaken all the communities of a layer to promote the detection on other layers.
    \item When detecting the local communities, we design a modified local spectral method as a subroutine, presenting an elaborate seed set augmentation process and a new community truncation rule.
    \item Extensive experiments on six synthetic datasets and six real-world datasets demonstrate the advantages of the proposed algorithm, which outperforms the referenced methods for multiple local community detection or global hidden community detection.
\end{itemize}

\section{Related Work}
To our knowledge, there is few literature addressing the local hidden community detection problem. However, several related research directions have already been studied for years. In this section, we briefly introduce three kinds of community detection that share some commonalities with the problem we present in this work.%, as well as a number of representative works of them.

\subsection{Multiple Local Community Detection}

Local community detection is a hot topic for social network analysis, whose purpose is to find a community supervised by a small seed set. Various methods have been proposed to handle this problem, including methods seeking for cohesiveness-optimization~\cite{yang2015defining,barbieri2015efficient,clauset2005finding}, local spectral methods~\cite{mahoney2012local,li2015uncovering,he2019krylov} and diffusion based methods~\cite{andersen2006local,kloster2014heat,bian2017many}.

Recently, some researchers began to address the problem of multiple local community detection, in which a seed belongs to a number of overlapping communities, and the goal is to uncover all these local communities. He \etal ~\cite{he2015detecting} first address this problem by proposing the M-LOSP algorithm that removes the seed from its ego network and regards the connected components of the ego network plus the original seed as new seed sets to find different local communities. Ni \etal ~\cite{ni2019local} propose an algorithm that selects seed sets from candidates near the seed and conducts local community detection separately to uncover local communities. S-MLC~\cite{kamuhanda2020sparse} applies a sparse non-negative matrix factorization to estimate the number of communities in the sampled subgraph and assign nodes to different communities. Although these algorithms can discover multiple local communities supervised by the given seed, they %only consider a single layer in the network. 
do not consider the existence of hidden structures. 
From a deeper perspective, our algorithm takes into consideration the different types of relations in the network and uncover multiple local communities of various strengths.
% among which some are dominant and others are hidden.

%\subsection{Community Detection on Networks with Multiple Layers}
\subsection{Hidden Community Detection}

The phenomenon of different community strengths and similar concepts to the layer in this paper are beginning to get noticed, for which some scholars have already designed a series of algorithms. Yang \etal ~\cite{young2015shadowing} discuss the fact that denser communities can overshadow sparser ones within a network when they have different resolutions. They conduct a cascading process, removing all the interior edges of the discovered communities to identify new communities. He \etal \cite{he2015revealing, he2018hidden} formally define the problem of ``hidden community detection'' and propose the HICODE algorithm to uncover the communities with different strengths. 
% Three reducing methods of community structures are optional, among which ReduceEdge and ReduceWeight maintain a community's density at the same level as its surroundings.
The HICODE algorithm takes into consideration the effect from hidden layers to dominant layers. Gong \etal ~\cite{gong2018finding} term the layers of a network ``multi-granularity'', and use embedding methods to discover and weaken the community structures. Although these works detect communities separately based on different layers, they concentrate on obtaining a global partition. To our knowledge, few work is presented to handle the hidden communities from the local perspective.

\subsection{Community Detection on Multilayer Networks}
There have also been works focusing on multilayer networks and address various kinds of tasks on them~\cite{tang2009uncoverning,chen2019tensor,interdonato2017local}. 
In multilayer networks, a layer contains a part of nodes and their interactions from the original network. Different layers are divided based on the interaction type, and there also exist inter-layer edges~\cite{kivela2014multilayer,huang2021survey}. To solve the community detection problem on multilayer networks, this series of work usually gets one partition of the network or communities with slight degree of overlapping after analysing different layers jointly. Note that although these researches also divide the network structure into several layers, the definition of layer is totally different from that in this paper. The ``layer'' we set is a partition of communities including all members in the network. Different layers correspond to different community division, and they have no relationship with each other. 
While the algorithms designed for multilayer networks do not consider the situation that members in the network may simultaneously belong to several groups with varying scales and cohesiveness, and should be detected separately.

\section{Preliminaries}
In this section, we first provide the problem definition, and then briefly introduce the local special (LOSP)~\cite{he2015detecting} method and a community structure weaken strategy called ReduceWeight~\cite{he2018hidden}.

\subsection{Problem Definition}
Let an undirected graph $G=(V, E, W)$ represent a network with $n=|V|$ nodes and $m=|E|$ edges. $W = \{w_{uv}|(u,v) \in E\}$ denotes the set of the weights for edges. Note that an unweighted network could be transformed into a weighted network by setting all the edge weights to be 1.
A layer $L$ of the network corresponds to a partition, which divides the network into a set of non-overlapping communities $L = \{C_h|h=1,2,\cdots\}$ covering all the nodes, \ie, $\forall v \in V$, ${\exists}C_h \in L$ \st $v \in C_h$. For a network containing more than one layers, we measure their strengths by calculating the community scoring function (\eg, modularity~\cite{newman2004finding}, conductance~\cite{shi2000normalized} or cut ratio~\cite{fortunato2010community}) of each partition, and consider the layer with comparatively highest score as the dominant layer and the others the hidden layers. 

Then, the local hidden community detection problem is defined as follows.

\textbf{Definition 1 (Local Hidden Community Detection).}  Given a graph $G$ including $n_L$ layers and a seed node $s$, the local hidden community detection problem aims to uncover the communities $\mathcal{C} = \{C^i|i=1,2,\cdots, n_L\}$,
% $\mathcal{C}=\{C_1, C_2, \cdots, C_{n_L}\}$, 
where $C^i$ is a community in the ${i}^{th}$ layer and $s \in C^i$ .

Note that although we term the problem ``local hidden community detection'', we focus on not only the communities containing the seed node in the hidden layers, but the communities in all layers, including the dominant one. The community in the dominant layer may have higher cohesiveness, but it is also interfered by a small number of interactions from the hidden layers, which affects the accurate detection of the dominant structure.

\subsection{The Local Spectral Method} \label{LOSP}
LOSP~\cite{he2015detecting,he2019krylov} is an effective method for local community detection based on the local spectral subspace. For a seed set $S$, LOSP performs a breadth-first-search (BFS) to sample a subgraph $G_S=(V_S, E_S, W_S)$ that covers the nodes closing to the seed members. To avoid fast expansion, LOSP removes the nodes with low inward ratio (the fraction of interior edges in the subgraph to the out-degree). If the number of sampled nodes exceeds a given threshold, a short random walk is conducted to filter the nodes to guarantee a proper size of the subgraph.

In the subgraph, LOSP performs a light lazy random walk to approximate the eigenvalue decomposition of the transition matrix and acquires a set of local spectral basis $\mathbf{V}_d^{(k)}$. Then, LOSP detects the local community by solving a linear programming problem: 
\begin{equation}
\begin{aligned}
&\text{min} \  \ ||\mathbf{y}||_1 = \mathbf{e}^{\text{T}} \mathbf{y} \\
\st \ &
(1)~~\mathbf{y} = \mathbf{V}_d^{(k)} \mathbf{u}, \\
&(2)~~\mathbf{y} \geq 0, \\
&(3)~~y_i \geq \frac{1}{|S|}, ~~i \in S. \\
\end{aligned}
\label{eq:LP}
\end{equation}
and finds a vector $\mathbf{y}$ in the local spectral subspace, in which $y_i$ indicates the likelihood of node $i$ belonging to the same community with the seed set. Finally, LOSP sorts all nodes in the subgraph according to their corresponding values in $\mathbf{y}$ in the descending order, and selects the top $|C_0|$ nodes as the output if the size of the target community $C_0$ is known. Otherwise, it automatically truncates the community based on the conductance score~\cite{shi2000normalized}.

He \etal \cite{he2019krylov} present some variants of the LOSP algorithm, including different types of random walks and another definition of local spectral subspace using the normalized adjacency matrix. For brevity, we only introduce the model used in our method, which has the best performance in experiments.

\subsection{The ReduceWeight Strategy} \label{reduce}
ReduceWeight is an effective strategy to reduce the connectivity inside communities as utilized in HICODE~\cite{he2018hidden}. 
Assuming each layer of the network is a separate stochastic blockmodel, in which every community is a block with its own edge density. The edges generated from other layers can be considered as background noises in the current model, ReduceWeight reduces the edge weights within the communities to weaken the densities of communities to the same degree with noises.

The interior density $p$ of a community $C$ with $n_C$ nodes is calculated from the weight sum of all edges in $C$ divided by the maximum possible number of edges (with weight 1), \ie,
\begin{equation}
	p = \frac{w_{C_{in}}}{\frac{1}{2} n_C(n_C-1)},
	\label{eq:p}
\end{equation}
where $w_{C_{in}}$ denotes the weight sum of the edges within $C$. Similarly, the background noise density $q$ is the ratio of actual exterior connection intensity to that of a full-connected graph:
\begin{equation}
	q = \frac{w_C - 2w_{C_{in}}}{n_C(n-n_C)},
	\label{eq:q}
\end{equation}
where $w_C$ is the weight sum of the edges which have at least one endpoint inside $C$. And the edge weights in the community can be reduced by multiplying the ratio of the two densities, such that
\begin{equation}
	w_{uv} = w_{uv} \cdot \frac{q}{p}, \  \ (u,v) \in E_C.
	\label{eq:weaken}
\end{equation}

For the reduced network with new weights, all the blocks disappear in the current blockmodel. In the meantime, the community structures in the hidden layer emerge. %are revealed.

\section{The Proposed Algorithm}
\label{sec:method}
In this section, we introduce the proposed Iterative Reduction Local Spectral (IRLS) algorithm in detail. 

\begin{figure}
    \centering
    \includegraphics[width=0.6\linewidth]{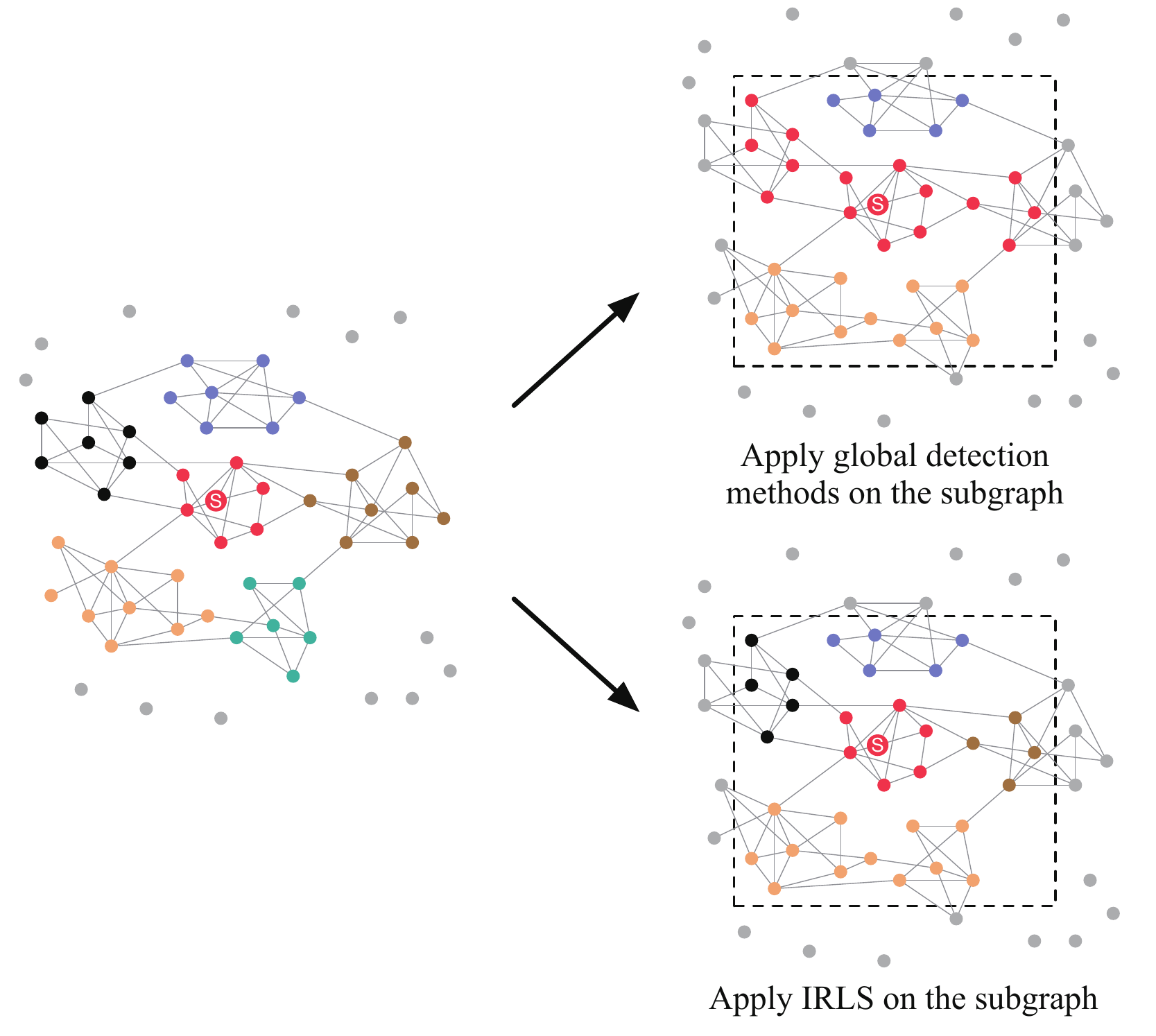}
    \caption{Different community detection results of one layer in the subgraph using two kinds of methods. Nodes of the same color (except for grey nodes abandoned in the sampling process) belong to the same community.}
    \label{fig:diff}
\end{figure}

\subsection{General Idea}

We first apply a sampling operation based on BFS if the network is very large, and subsequent processes are all executed iteratively in the sampled subgraph. The sampling process gives priority to reserve the local communities containing the seed, but in terms of other communities surrounding them, they may not be intact in the subgraph. And these ``broken'' communities (a community is broken when a portion of its members is reserved in the subgraph and another portion is abandoned) can affect the overall structure of the network. Under this circumstance, when we apply some global detection methods in the subgraph, these communities may not be identified separately and be ``absorded'' by other larger communities, badly interfering the detection results, as shown in Figure \ref{fig:diff}.

To avoid this situation as much as possible and guarantee the detection accuracy of the most valuable local communities, in our algorithm we divide the detection operation of each layer into two stages to detect the local community and the neighborhood communities in turn. At each iteration, we adopt the LOSP algorithm with a seed set augmentation to find a local community containing the seed node. Then we temporarily remove all the node members of this community as well as the edges connecting with them and detect other communities in the remaining subgraph. Now we merge the detected communities as a partition of the subgraph. Following the ReduceWeight strategy, the connections inside the communities in a layer are weakened so that the current layer's structure is weaker than others, and we could repeat the detection process to explore another layer again.
In this way, the local communities containing the seed node are uncovered in all layers.

\subsection{Subgraph Sampling}
For some networks in large scale, it takes high computation overhead and space cost to load data and perform algorithms on the networks. Furthermore, not the overall network is required when we attempt to address problems from the local perspective. In the local community detection, the members belonging to the target communities are usually at a short distance from the seed in the graph, which means the paths connecting the seed with them consist of only one or several edges. Therefore, we use the sampling process described in Section \ref{LOSP} when a network is very large, so as to reserve the nodes and edges that are helpful for detecting and abandon the irrelevant information. Note that during our iterative algorithm, while we perform the detection process for multiple times, the sampling is applied only once at the beginning, and the remaining operations are all conducted in the sampled subgraph.

\subsection{Local Community Detection} \label{local}

To find out the community containing the seed node in each layer, we design a modified local spectral algorithm, which makes two alterations to the existing LOSP to make it more suitable for the local hidden community detection problem, \ie, the community boundary determination and the process of seed set augmentation. 
Details of the modified LOSP are presented in Algorithm \ref{alg:LOSP}.

\begin{algorithm}
\begin{algorithmic}[1]
\caption{The Modified Local Spectral Method}
\label{alg:LOSP}
\REQUIRE Graph $G$, seed node $s$.
\ENSURE Local community $C_0$.
\renewcommand{\algorithmicrequire}{\textbf{Parameters:}}
\REQUIRE Maximal size of seed set $n_{\text{set}}$, maximal size of community $n_{\text{com}}$, value threshold $\beta$, ratio threshold $\gamma$.

\STATE $G_{S} \leftarrow \text{advanced BFS sampling}(G)$
\STATE $S \leftarrow \{s\}$
\WHILE{$|S| \leq n_{\text{set}}$}
    \STATE obtain $\mathbf{y}$ by solving Eq. \eqref{eq:LP} and sort $\mathbf{y}$ in the descending order, $\mathbf{v}$ is the corresponding node sequence
    \IF{$\exists \; j,k \in \{y_1,...,y_{|S|}\}$, \st $y_j \geq 2y_k$}
        \STATE \textbf{break}
    \ENDIF
    \FOR{$j=1:2*n_{\text{set}}$}
        \IF{$y_j \geq \beta \text{ or } y_j/y_{j+1} \geq \gamma$}
            \STATE $S \leftarrow$ \{$v_1,...,v_j$\}
        \ENDIF
    \ENDFOR
    \IF{$|C_0|$ is known}
        \STATE $C_0 \leftarrow \{v_1,...,v_{|C_0|}\}$
        \ELSE \FOR{$j=1:n_{\text{com}}$}
        \STATE $C_j \leftarrow \{v_1,...,v_j\}$
        \STATE calculate $Q(C_j)$ by Eq. \eqref{eq:modu}
        \ENDFOR
        \STATE $C_0 \leftarrow C_{j^*} \text{ with maximal } Q(C_{j^*})$
    \ENDIF
\ENDWHILE

\end{algorithmic}
\end{algorithm}

\subsubsection{Community Boundary Determination} \label{comsize}
When the actual size of the target community $|C_0|$ is unknown, LOSP seeks for a local minimum of conductance score to decide the community boundary. For a node set $C$ and its complement set $\overline{C} =  V \setminus C$ in $G$, the conductance of $C$ is defined by:
\begin{equation}
    \phi(C) = \frac{\text{cut}(C, \overline{C})}{\text{min}\{d_C, d_{\overline{C}}\}},
\end{equation}
where cut(·, ·) denotes the number of edges between two node sets, and $d$ denotes the total degrees of all nodes in a node set.
For index $j$ of the sorted node sequence $\mathbf{v}$ in LOSP, the conductance of the first $j$ nodes $\phi(C_j)$ is calculated, starting from an index that the node set ahead contains all the seeds. When $\phi(C_{j^*})$ reaches a local minimum and then keeps increasing until it surpasses 1.02$\phi(C_{j^*})$, the first $j^*$ nodes are considered as the detected community. However, in the process of calculating the conductance, we need to count the total degrees for both node sets and the edges in their junction whenever we add one node to the community, which is costly in computation. And this function is not applicable to weighted networks. In order to reduce the time complexity and improve the algorithm's %universality
generality, we introduce the weighted local modularity as the scoring function.

\textbf{Definition 2 (Weighted Local Modularity).} For a node set $C$, the weighted local modularity $Q(C)$ is defined by:
\begin{equation}
    Q(C) = \frac{\frac{w_{Cin}}{w_G} - (\frac{w_C}{2w_G})^2}{n_C},
    \label{eq:modu}
\end{equation}
where $w_G$ denotes the weight sum of edges in the graph. 

% The traditional modularity $Q = \sum_i n_{C_i}\cdot Q(C_i)$ is designed for the global partition of a graph, while the measurement in Eq. \eqref{eq:modu} is more suitable to evaluate the quality of a local community.
When calculating the weighted local modularity, we only need to focus on the community itself without its complement set, omitting a large amount of computation. With this improvement, we are able to calculate $Q(C)$ from the starting index to the ceiling size of a community and pick the global optimum, rather than stopping at the point with a local minimal conductance. Additionally, the Louvain method~\cite{blondel2008fast} is used to detect other communities in the remaining subgraph in Section \ref{othercom}, which is based on the modularity. Using weighted local modularity can also ensure the compatibility for the whole algorithm.

\subsubsection{Seed Set Augmentation}
For the basic local community detection, a comparatively larger seed set usually leads to a more accurate detection. In the problem proposed in this paper, starting from only one single seed node is not convincing to discover the complete community structure. As the linear programming solution $\mathbf{y}$ indicates the nodes' probabilities that they belong to the target community, we can use the values of this vector to pick nodes for expanding the seed set. A naive method is presented in~\cite{he2015detecting}, which takes a periodic increase in the size of the seed set, selecting new seeds from the top of $\mathbf{v}$ simply by a fixed number at each iteration, and uses the expanded seed set as a new input to perform the algorithm again. %It's easy to recognize the disadvantages of this method. 
However, when adding new nodes to the seed set, we do not know to what extent a node is likely to be in the same community as the current seeds, and if a node from another community is chosen, the results of detection will be affected and decay. 

To address this issue, we set two conditions to evaluate potential seed nodes. First, if a node $v_j$'s corresponding value in $\mathbf{y}$ is larger than a threshold $\beta$, we consider that this node has a high likelihood to belong to the target community, and it can be regarded as a seed. Second, if $v_j$'s value is $\gamma$ times that of the next node in $\mathbf{v}$, we take them as the possible boundary of two communities, and the top nodes up to $v_j$ can all be qualified as members in the new seed set. We test with several values of $\beta$ and $\gamma$ and eventually set $\beta = 0.3, \gamma = 1.05$ for all the datasets in experiments. The conditions improve the quality of new seeds, though they still cannot ensure an absolutely correct augmentation process. Therefore, we also present a ``revocation mechanism''. During an iteration with a newly-expanded seed set, if in the recalculated $\mathbf{y}$ 
%there is a gap of more than twice that in all the seeds' likelihood value 
the lowest value is no greater than a half of the highest value, indicating the values vary significantly, then there may exist some wrong nodes from other communities that should not be included in the seed set. Then we stop the algorithm and use the community detected in the previous iteration as the output.

We set an upper bound of the seed set size, and the augmentation process terminates when the seed set reaches the upper bound. In Section \ref{exp}, we show some experiments and observations about different values about the maximal seed set size.

\subsection{Neighborhood Community Detection} \label{othercom}
In a social network with structures of different cohesiveness, an individual usually has several kinds of relationships correspondingly. The interactions from the dominant layer are usually stronger, and we can consider that the hidden layers are ``covered'' by the dominant one. If we want to discover the hidden communities, the interference from the stronger communities must be eliminated. Most of the time, as the communities from different layers are not highly overlapping with each other, parts of a few dominant communities may form the cover over a weaker one together. As shown in Figure \ref{fig:network}, The members of a community in the hidden layer are divided into parts and belong to four dominant communities, which have to be detected and weakened if we want to uncover the hidden group. In Section \ref{local}, we detect one community containing the seed node in one layer (for the original network, the first time of detection always tends to discover the community in the dominant layer), but this is not enough to reveal the hidden structure. 

Therefore, apart from the local community, we also detect other communities for the remaining part of the network. Louvain method~\cite{blondel2008fast} is a popular global community detection method based on modularity maximization. After obtaining the output $C_0$ in Section \ref{local}, we remove the members in the community as well as all edges connecting with them from the subgraph, and apply Louvain method on the remaining subgraph to uncover the neighborhood communities $C_1,C_2,...,C_h$. Then the communities detected by both steps compose the partition of the current layer in the subgraph, \ie, $\mathcal{C}_L = \{C_i | i=0,1,...,h\}$.

\begin{figure}
    \centering
    \includegraphics[width=0.3\linewidth]{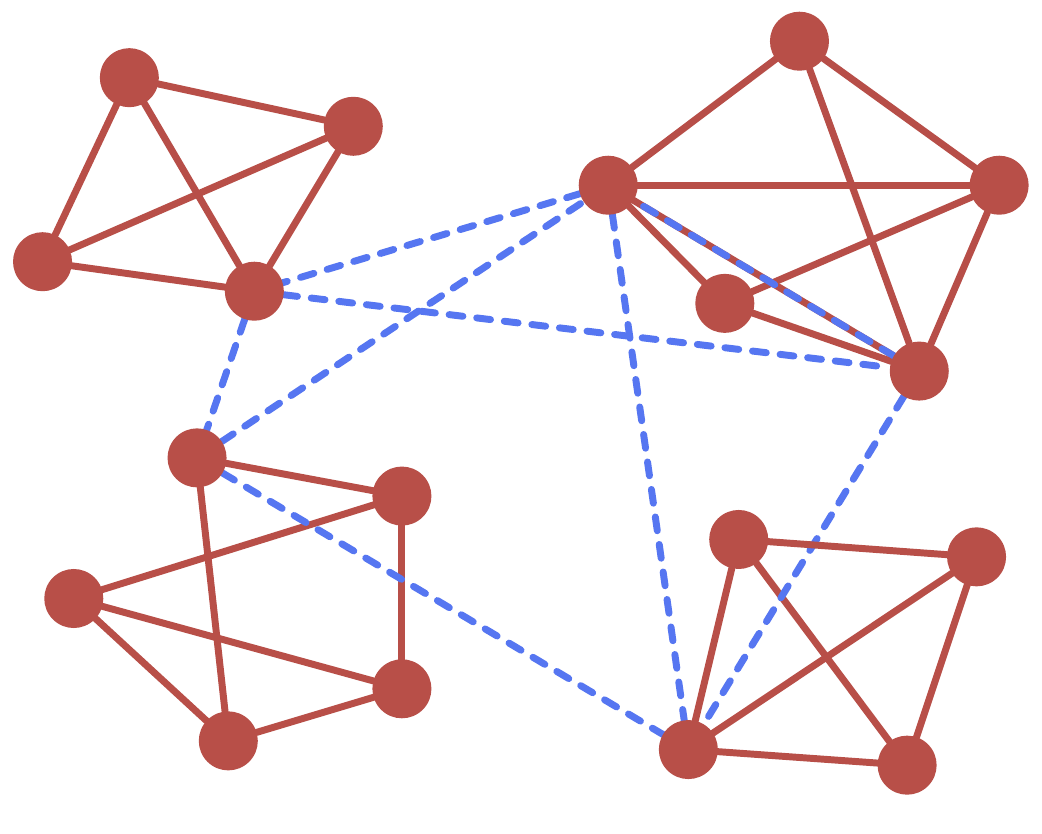}
    \caption{An example of a network with two layers. There are four communities represented by red solid lines in the dominant layer and one community represented by blue dashed lines in the hidden layer.}
    \label{fig:network}
\end{figure}

Note that although the obtained partition seems to have the same form as generated directly by a global community detection algorithm, the results of the two ways usually differ in quality, especially when the network does not have strong community structures, or the sampling process breaks a number of communities. Instead of the accurate detection of %other communities used for weakening a layer,
all communities, 
we are most concerned about the detection accuracy of the community containing the seed, to which our method gives priority, while global algorithms focus more on a balanced partition with an overall optimal performance.

Consider the example shown in Figure \ref{fig:diff}, several communities are broken after sampling and the nodes in the subgraph also have some connections outside their group. If we use a global detection method, the reserved nodes may be misjudged and clustered into other communities, including the target local community. And if we begin with the seed to discover the local community, we can at least guarantee this step of detection leads to the best outcome. The individuals who lose their partners can also be determined correctly in this way, and even if they are contained in other groups, the little error in the weakening stage affects the detection accuracy much less than the incorrect detection itself. This example explains the superiority of our algorithm over the global hidden community detection methods. In Section \ref{Theory}, we will further discuss the difference theoretically.
% Comparative experiments are conducted in Section \ref{exp} to confirm this difference.

% \renewcommand{\algorithmicrequire}{\textbf{Input:}}
% \renewcommand{\algorithmicensure}{\textbf{Output:}}
\begin{algorithm}
\begin{algorithmic}[1]
\caption{The IRLS Algorithm}
\label{alg:whole}
\REQUIRE Graph $G$, seed node $s$, number of layers $n_L$.
\ENSURE Local communities $\mathcal{C}_0$.
\renewcommand{\algorithmicrequire}{\textbf{Parameters:}}
\REQUIRE Number of iterations $T$.

\STATE sample $G_{S}=(V_S, E_S, W_S)$ by line 1 in \textbf{Algorithm \ref{alg:LOSP}}
\FOR{$t = 1 : T$}
    \STATE $G_{0}=(V_0, E_0, W_0) \leftarrow G_{S}=(V_S, E_S, W_S)$ 
    \FOR{$i = 1 : n_L$}
	  \item
        $L_i \leftarrow \emptyset$
    \ENDFOR
    \FOR{$i = 1 : n_L$}
        \FOR{$j = 1 : n_L, j \neq i$}
            \IF{$L_j \neq \emptyset$}
                \STATE weaken each community in $L_j$ by Eq. \eqref{eq:p} - \eqref{eq:weaken}
            \ENDIF
        \ENDFOR
        \STATE detect $C_0^i$ by line 2 - 22 in \textbf{Algorithm \ref{alg:LOSP}}
    \STATE $V'_0 \leftarrow V_0 \setminus C_0^i$, $G'_0 \text{ is induced from } V'_0$
    \STATE $\{C_1^i,...,C_h^i\} \leftarrow \text{Louvain}(G'_0)$
    \STATE $L_i \leftarrow \{C_0^i,...,C_h^i\}$
    \ENDFOR
\STATE $\mathcal{C}_0 \leftarrow \{C_0^1, C_0^2,...,C_0^{n_L}\}$
\ENDFOR
\end{algorithmic}
\end{algorithm}

\subsection{Layer Reduction and Community Refinement}
After all communities in the dominant layer are detected, we calculate the interior and exterior densities of each community and use the ReduceWeight method to weaken them and make the underneath hidden structure evident. Then the operations above are repeated to discover a new partition.

While the communities in the dominant layer cover that in other layers, the structures of the latter have an impact on the former as well. Therefore, after discovering the hidden communities, we also weaken them for the dominant communities' detection, and apply the algorithm iteratively to refine the results of all layers. For each iteration (the subgraph and seed set keeps original at the beginning of each iteration), if the partitions of the other layers are already known, we weaken their structures, then apply the operation in Section \ref{local} and Section \ref{othercom} to detect the local community and neighborhood communities. And the detection result of the current layer is updated. Eventually, the refined communities containing the seed in all dominant and hidden layers are detected after a number of iterations. The details of IRLS are presented in Algorithm \ref{alg:whole}. 

\subsection{Theoretical Analysis} \label{Theory}
In this subsection, we demonstrate the necessity of using the IRLS method to solve the local hidden community detection problem from the theoretical perspective. We prove that %when surrounded by incomplete neighborhood communities, 
under certain conditions, the detection of the dominant local community using global hidden community detection algorithms can be interfered, while the IRLS can avoid such inaccuracy. And this also contributes to the detection of the hidden layer underneath.

The following proof is given on a two-layer network, where each layer is a stochastic blockmodel, and layer 1 and layer 2 are the dominant layer and the hidden layer, respectively. For simplicity, we make the assumption that every community (corresponding to a block) in one layer has the same size and same probability of forming edges with weight 1, and the edges can only be generated between two nodes from the same community including self-loops. For each community in layer $l$, $n_l$, $e_{lin}$ and $e_{lout}$ denote the number of nodes and edges inside and edges that have one endpoint inside, respectively.

As mentioned in Section \ref{othercom}, when we execute the sampling process based on BFS from the seed node, the primary goal is to reserve the complete local community containing the seed in the dominant layer, and we cannot prioritize the  surrounding communities. Some of these neighborhood communities may be located on the boundary of the subgraph, and a number of their members are abandoned through the sampling. These left portions increase the difficulty of the detection. In Theorem \ref{thm1}, we derive the conditions that a broken neighborhood community (the situation about multiple neighborhood communities can be calculated in a similar way) is merged into the intact local community by the representative global hidden community detection method HICODE, while IRLS can distinguish them. We use $C_1$ and $C_2$ to denote the local community and the neighborhood community. $n_{C_i}$ is the number of nodes in the community $C_i$. Assume that $t$ proportion of $C_2$ has been reserved in the subgraph, \ie, $n_{C_2} = tn_{C_1}$ \footnote{For simplicity, we make the assumption that the edges inside and outside the communities are uniformly distributed with the nodes.}. And $e$ and $e_{12}$ denote the number of edges in the whole subgraph and edges between $C_1$ and $C_2$ (one endpoint belongs to $C_1$ and one endpoint belongs to $C_2$), respectively.

\newtheorem{thm}{\bf Theorem}
\begin{thm}
If $t$ proportion of $C_2$ is reserved in the subgraph. When $t$ satisfies the following two conditions, HICODE would consider $C_1$ and $C_2$ as one community while IRLS would detect them separately, as: 
\begin{equation}
    (2e_{1in}^2 + e_{1in}e_{1out})t + \left(e_{1in}e_{1out}+\frac{1}{2}e_{1out}^2-\frac{n_1ee_{1out}}{n-n_1}\right) < 0,
\label{eq:inequ1}
\end{equation}
\begin{equation}
    e_{1in}^2t^3 + (e_{1in}e_{1out}-ee_{1in})t^2 + \left(2e_{1in}^2+e_{1in}e_{1out} + \frac{1}{4}e_{1out}^2+ee_{1out}\right)t + \left(\frac{1}{4}e_{1out}^2-e_{1in}^2-\frac{n_1ee_{1out}}{n-n_1}\right) > 0.
\label{eq:inequ2}
\end{equation}
\label{thm1}
\end{thm}

\begin{proof}
For the convenience during the calculation of modularity, we present this proof on an unweighted graph and adapt the weighted local modularity to:
\begin{equation}
    Q(C_i) = \frac{\frac{e_{C_iin}}{e} - \left(\frac{d_{C_i}}{2e}\right)^2}{n_{C_i}},
    \label{eq:modu1}
\end{equation}
where $d_{C_i} = 2e_{C_iin} + e_{C_iout}$,
% is the summation of all nodes' degrees in community $C_i$, 
and $e_{C_iin} \text{ and } e_{C_iout}$ denote the numbers of internal and external edges of community $C_i$.
When $n_{C_2} = tn_{C_1}$, due to the uniform distribution of the edges, we have $e_{C_2in} = t^2e_{C_1in}, e_{C_2out} = te_{C_1out}$. And if we consider $C_1$ and $C_2$ together as a community, denoted as $C_{12}$, we have $n_{C_{12}} = (t+1)n_{C_1}, e_{C_{12}in} = (t^2+1)e_{C_1in} + e_{12}, e_{C_{12}out} = (t+1)e_{C_1out} - 2e_{12}$.

HICODE uses the Louvain method to detect the partition of a layer, which is based on the modularity maximization.

\begin{equation}
    Q'(l) = \sum_i Q'(C_i) = \sum_i \left(\frac{e_{C_iin}}{e} - \left(\frac{d_{C_i}}{2e}\right)^2\right).
\end{equation}

For $C_1$ and $C_2$, no matter how they are connected with other communities, as long as the modularity of $C_{12}$ is higher than the modularity summation of $C_1$ and $C_2$, \ie, $Q'(C_{12}) > Q'(C_1) + Q'(C_2)$, the two communities cannot be detected separately by HICODE.

When $Q'(C_{12}) > Q'(C_1) + Q'(C_2)$, we have:
\begin{equation}
\begin{aligned}
    & Q'(C_{12}) - (Q'(C_1) + Q'(C_2)) > 0 \\
    & \iff \left(\frac{e_{C_{12}in}}{e} - \left(\frac{d_{C_{12}}}{2e}\right)^2\right) - \left(\left(\frac{e_{C_1in}}{e} - \left(\frac{d_{C_1}}{2e}\right)^2\right) +      \left(\frac{e_{C_2in}}{e} - \left(\frac{d_{C_2}}{2e}\right)^2\right)\right) > 0
    \\
    & \iff \left(\frac{(t^2+1)e_{C_1in} + e_{12}}{e} - \left(\frac{(2t^2+2)e_{C_1in}+(t+1)e_{C_1out}}{2e}\right)^2\right)\\
    & - \left(\left(\frac{e_{C_1in}}{e} - \left(\frac{2e_{C_1in}+e_{C_1out}}{2e}\right)^2\right) +       \left(\frac{t^2e_{C_1in}}{e} - \left(\frac{2t^2e_{C_1in}+te_{C_1out}}{2e}\right)^2\right)\right) > 0
    \\
    & \iff \left(\frac{(t^2+1)e_{1in}}{e} - \frac{(4t^2+4)e_{1in}^2+(4t^3+4)e_{1in}e_{1out}+(t^2+1)e_{1out}}{4e^2}\right) \\
    & - \left(\frac{(t^2+1)e_{1in}+e_{12}}{e} - 
    \frac{(4t^4+8t^2+4)e_{1in}^2+(4t^3+4t^2+4t+4)e_{1in}e_{1out}+(t^2+2t+1)^2e_{1out}^2}{4e^2}\right) > 0 \\
    & \iff ee_{12} > 2t^2e_{1in}+(t^2+t)e_{1in}e_{1out}+\frac{1}{2}te_{1out}^2.
\end{aligned}
\label{eq:inequ3}
\end{equation}

On the contrary, we want the weighted local modularity of $C_1$ to be higher than that of $C_{12}$, \ie, $Q(C_1) > Q(C_{12})$, so that IRLS would not combine $C_1$ and $C_2$ and truncate the right community.

When $Q(C_1) > Q(C_{12})$, we have:
\begin{equation}
    \begin{aligned}
    & Q(C_1) - Q(C_{12}) > 0 \\
    & \iff \frac{\frac{e_{C_1in}}{e} - \left(\frac{d_{C_1}}{2e}\right)^2}{n_{C_1}}
    - \frac{\frac{e_{C_{12}in}}{e} - \left(\frac{d_{C_{12}}}{2e}\right)^2}{n_{C_{12}}} > 0
    \\
    & \iff \frac{\frac{e_{C_1in}}{e} - \left(\frac{2e_{C_1in}+e_{C_1out}}{2e}\right)^2}{n_{c_1}}
     - \frac{\frac{(t^2+1)e_{C_1in} + e_{12}}{e} - \left(\frac{(2t^2+2)e_{C_1in}+(t+1)e_{C_1out}}{2e}\right)^2}{(t+1)n_{c_1}} > 0 \\
    & \iff \left(\frac{e_{1in}}{n_1e}-\frac{4e_{1in}^2+4e_{1in}e_{1out}+e_{1out}^2}{4n_1e^2}\right) \\
    & - \left(\frac{(t^2+1)e_{1in}+e_{12}}{(t+1)n_1e} - 
    \frac{(4t^4+8t^2+4)e_{1in}^2+(4t^3+4t^2+4t+4)e_{1in}e_{1out}+(t^2+2t+1)^2e_{1out}^2}{4(t+1)n_1e^2}\right) > 0 \\
    & \iff ee_{12} < (t^4+2t^2-t)e_{1in}^2 + (t^3+t^2)e_{1in}e_{1out}+\left(\frac{1}{4}t^2+\frac{1}{4}t\right)e_{1out}^2+(t-t^2)ee_{1in}.
    \end{aligned}
\label{eq:inequ4}
\end{equation}

Let $b_l$ and $p_l$ respectively denote the number of communities and the probability of forming edges in layer $l$. Since the edges counted by $e_{12}$ and $e_{1out}$ are both generated in layer 2, we have $e_{12} = \frac{tp_2n_1^2}{b_2}, e_{1out} = \frac{p_2n_1(n-n_1)}{b_2}$. So $e_{12}$ can be represented by $e_{1out}$:
\begin{equation}
    e_{12} = \frac{tn_1}{n-n_1}\cdot e_{1out}.
\label{eq:edges}
\end{equation}
Combining Eq. \ref{eq:inequ3} - \ref{eq:inequ4} and Eq. \ref{eq:edges}, we can obtain the two conditions that $t$ needs to satisfy:
\begin{equation*}
    (2e_{1in}^2 + e_{1in}e_{1out})t + \left(e_{1in}e_{1out}+\frac{1}{2}e_{1out}^2-\frac{n_1ee_{1out}}{n-n_1}\right) < 0,
\end{equation*}
\begin{equation*}
    e_{1in}^2t^3 + (e_{1in}e_{1out}-ee_{1in})t^2 + \left(2e_{1in}^2+e_{1in}e_{1out} + \frac{1}{4}e_{1out}^2+ee_{1out}\right)t + \left(\frac{1}{4}e_{1out}^2-e_{1in}^2-\frac{n_1ee_{1out}}{n-n_1}\right) > 0.
\end{equation*}

We use the online calculator Wolfram|Alpha \footnote{https://www.wolframalpha.com/} to solve the above two inequalities, %and we omit the results here because of their length.
and we omit the detailed calculations here. %results here because of their length.
Eq. \ref{eq:inequ1} is a linear inequality and easy to solve. Eq. \ref{eq:inequ2} is a cubic inequality, whose corresponding equation has a real root and two imaginary roots, so $t$ only needs to be greater than the real root to satisfy the inequality. 

\end{proof}

In~\cite{bao2020hidden}, the authors prove that for a two-layer block model network, the modularity of a layer increases if we apply ReduceWeight on all communities in the other layer. Now we show the different influences that the two situations mentioned in Theorem \ref{thm1} have to the followup detection.

\begin{thm}
The modularity increment of layer 2 after addressing ReduceWeight when two communities in layer 1 are considered as one is lower than that when the two communities are detected separately.
\label{thm2}
\end{thm}

\begin{figure}
    \centering
    \includegraphics[width=0.3\linewidth]{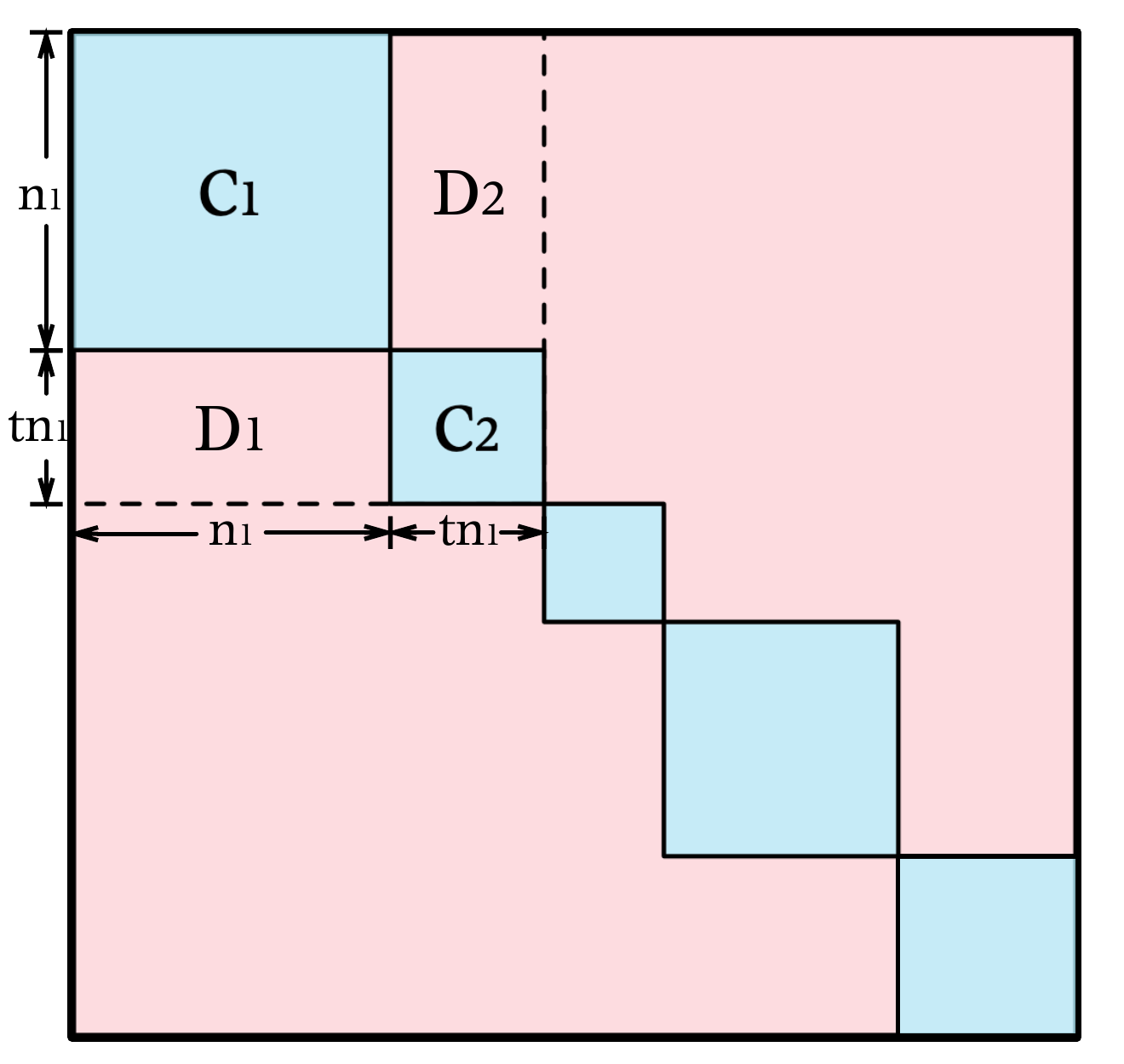}
    \caption{The blockmodel of the sampled subgraph from the perspective of layer 1. The probabilities of forming edges in the blue and red areas are $p_1$ and $p_2$, respectively.}
    \label{fig:block}
\end{figure}

\begin{proof}
In Figure \ref{fig:block}, we demonstrate the subgraph's blockmodel from the perspective of layer 1. The communities in layer 1 are represented by squares, whose observed probability of forming edges is $p_1$.
% and $p_1' = p_1 + p_2 - p_1p_2$ (some node pairs are in the same communities in both layer 1 and layer 2).
And for the rest part other than the squares, the edges generated by layer 2 are considered as background noises with probability $p_2$.

When $C_1$ and $C_2$ are detected separately, we use ReduceWeight strategy and multiply the interior edges' weights by $\frac{p_2}{p_1}$. The summation of the reduced edge weights in the two communities is:
\begin{equation}
    \Delta W^1 = (p_1-p_2)(t^2+1)n_1^2.
\end{equation}
Assuming that the edges generated by layer 2 account for $r$ ($0<r<1$) percent of all edges in $C_1$ and $C_2$, the edges inside the communities in layer 2 are reduced by:
\begin{equation}
    \Delta W_{2in}^1 = (p_1-p_2)(t^2+1)n_1^2 \cdot r.
\end{equation}
And the rest part of $\Delta W$ is considered as weakening the outgoing edges in layer 2:
\begin{equation}
    \Delta W_{2out}^1 = (p_1-p_2)(t^2+1)n_1^2 \cdot (1-r).
\end{equation}
% which can be considered as the probability of forming edges in the two communities are cut from $p_1$ to $p_2$. 

When the two communities are detected incorrectly, \ie, $C_1$, $C_2$, $D_1$ and $D_2$ are regarded as a community $C_{12}$ together, the edges in all the four areas are weakened in proportion to their numbers of edges to keep the probabilities of $C_{12}$ at $p_2$. In this case, the interior reduction in layer 2 includes the alteration in $D_1$ and $D_2$ and $r$ percentage of that in $C_1$ and $C_2$:
\begin{equation}
\begin{aligned}
    \Delta W_{2in}^2 = & (p_1-p_2)(t^2+1)n_1^2 \cdot \frac{2tn_1^2}{(t+1)^2n_1^2}
     + (p_1-p_2)(t^2+1)n_1^2 \cdot \frac{(t^2+1)n_1^2}{(t+1)^2n_1^2} \cdot r.
\end{aligned}
\end{equation}
And the reduction of outgoing edges in layer 2 is:
\begin{equation}
    \Delta W_{2out}^2 = (p_1-p_2)(t^2+1)n_1^2 \cdot \frac{(t^2+1)n_1^2}{(t+1)^2n_1^2} \cdot (1-r).
\end{equation}
Evidently, $\Delta W_{2in}^1<\Delta W_{2in}^2, \Delta W_{2out}^1>\Delta W_{2out}^2$. According to Lemma 2 in ~\cite{bao2020hidden}, if the layer weakening method reduces a bigger percentage of outgoing edges than internal edges in layer $l$, the modularity of layer $l$ increases after the weakening. Therefore, the modularity of layer 2 after weakening when $C_1$ and $C_2$ are detected separately is higher than that when the two communities are regarded as one:
\begin{equation}
    Q'_1(l_2)>Q'_2(l_2).
\end{equation}
And $\Delta Q'_1(l_2) = Q'_1(l_2) - Q'_0(l_2), \Delta Q'_2(l_2) = Q'_2(l_2) - Q'_0(l_2)$, where $ Q'_0(l_2)$ is the original modularity of layer 2. Hence, we have: 
\begin{equation}
    \Delta Q'_1(l_2)>\Delta Q'_2(l_2).
\end{equation}
\end{proof}

\section{Experimental Results} 
\label{exp}
We evaluate the performance of IRLS on both synthetic datasets and real world networks.
As IRLS is the first method to conduct the local hidden community detection, we select methods from related problem categories as the baselines: a global hidden community detection algorithm, HICODE~\cite{he2018hidden} and a multiple local community detection algorithm, M-LOSP~\cite{he2015detecting}. All experiments are conducted on a machine with two Intel Xeon CPUs at 2.3GHZ and 256GB main memory.

For a seed node, we adopt the sampling process described in Section \ref{local} to generate a subgraph, in which the number of BFS steps is set to three. Based on the observation that the community distribution in real-world networks is more compact than those generated randomly, the size threshold of synthetic datasets and real-world datasets are set 10,000 and 5,000, respectively. In IRLS, we set $T = 10, \beta = 0.3, \gamma = 1.05; n_{\text{com}} = 150, n_{\text{set}} = 18$ for synthetic datasets and $n_{\text{com}} = 500, n_{\text{set}} = 9$ for real-world datasets (the sizes of communities in the real-world networks range from dozens to thousands of 
nodes). The other parameters of the modified LOSP are consistent with the original LOSP algorithm in~\cite{he2019krylov}. We name the two types of IRLS methods truncated by ground truth size or weighted local modularity as IRLS-gt and IRLS-auto, respectively. For HICODE, we assume the number of layers is given, so we run its refinement stage only, and choose ReduceWeight as the weakening strategy. Other settings for the two baselines keep default or are set to the values with the best performance.

\subsection{Evaluation Metric}
We use F1-score as the evaluation metric to measure the performance of the algorithms. Given a ground truth community $C^*$ and a detected community $C$, the precision and recall are defined by:
\begin{equation*}
	precision = \frac{|C \cap C^*|}{|C|},
	recall = \frac{|C \cap C^*|}{|C^*|}.
\end{equation*}
Then, the F1-score is calculated as follows:
\begin{equation}
	F_1 = \frac{2 * precision * recall}{precision + recall}.
\end{equation}

Specifically, for the partitions produced by HICODE in the subgraph, we find the communities containing the seed node, and take them as the local communities in the corresponding layers. For M-LOSP, we detect local communities from $n_L$ largest seed sets, calculating their F1-scores with the ground truth communities in all layers, then assign them to the layers with the highest score. If more than one communities are assigned to a same layer, we take the one with the highest F1-score. If no community is assigned to a layer, the score of this layer is 0.

Additionally, to guarantee the validity of seed set augmentation process and the clear structure of multiple local communities, we take a node as a potential seed when its ground truth communities in all layers are of sizes higher than $n_{\text{set}}$, and contain at least one other node connecting with it. In each case, a seed is randomly selected from the potential seeds, and we run 100 cases for the four methods on all datasets and calculate the average results.

\subsection{Comparison on Synthetic Datasets}
As mentioned in Section \ref{reduce}, each layer of a network can be regarded as a stochastic blockmodel, with every community corresponding to a block. Based on this observation, we create six synthetic datasets with 30,000 nodes, which have two different types of community configurations.
Following the fact that the community sizes of real-world networks can be approximated by the power law distribution~\cite{lancichinetti2008benchmark}, we randomly pick each community's size between 30 to 100 from a power law with exponent $\tau=1$, determining the nodes' membership in each layer and generate SynL2\_1 with two layers and SynL3\_1 with three layers. For the other four synthetic networks, we specifically set the number of communities in each layer, and every node is randomly assigned to a community within a layer. Both situations that communities in different layers have various or roughly similar sizes are taken into consideration.

After forming the communities, all layers' respective $p$ values are set to produce a series of $G_C=(n_C, p)$ Erdos-Renyi random graphs for each block. In this way, we guarantee a network with multiple layer structures of various strengths. Additionally, considering  that the partitions cannot be completely independent from each other, we add some background noises by producing another Erdos-Renyi graph $G=(n, p_0)$ for the whole network. The statistic information of the synthetic datasets are shown in Table \ref{tab:synthetic}, where values for different layers are separated by ``;'' in one grid, and ``\# Communities'' indicates the number of communities in each layer.

\begin{table*}
	\scriptsize
	\renewcommand{\arraystretch}{1.3}
	\caption{Statistics on synthetic datasets.}
	\label{tab:synthetic}
	\centering
	\resizebox{\textwidth}{15mm}{
	\begin{tabular}{l l l l l l l l l}
		\toprule
		Dataset  & $n$    & $m$   & $n_L$ & Modularity         & \#Communities & $\overline{|C|}$ & $p$ & $p_0$\\
		\midrule
		SynL2\_1 & 30,000  & 885,288    & 2 & 0.271; 0.219        & 518; 509     & 57.92; 58.94   & 0.25; 0.20  & 0.001\\
		SynL2\_2 & 30,000 & 973,678   & 2 & 0.307; 0.229        & 600; 300   & 50; 100  & 0.40; 0.15  & 0.001\\
		SynL2\_3 & 30,000 & 897,522   & 2 & 0.300; 0.199        & 500; 500   & 60; 60   & 0.30; 0.20   & 0.001\\
		SynL3\_1 & 30,000  & 1,023,980    & 3 & 0.234; 0.185; 0.140 & 516; 524; 519    & 58.14; 57.25; 57.80 & 0.25; 0.20; 0.15 & 0.001\\
		SynL3\_2 & 30,000 & 1,004,598 & 3 & 0.225; 0.178; 0.148 & 1,000; 500; 300 & 30; 60; 100 & 0.50; 0.20; 0.10 & 0.001\\
		SynL2\_3 & 30,000 & 987,188   & 3 & 0.227; 0.181; 0.136 & 500; 500; 500 & 60; 60; 60 & 0.25; 0.20; 0.15 & 0.001\\
		\bottomrule
	\end{tabular}}
\end{table*}

\makeatletter
\def\hlinew#1{%
	\noalign{\ifnum0=`}\fi\hrule \@height #1 \futurelet
	\reserved@a\@xhline}
\makeatother

\begin{table*}
	\scriptsize
	\renewcommand{\arraystretch}{1.3}
	\caption{Comparison of F1-score on synthetic datasets.}
	\label{tab:SYNresults}
	\centering
	\resizebox{\textwidth}{15mm}{
	\scalebox{1}{
		\begin{tabular}{l| l l| l l| l l| l l}
			\hline
			% \toprule
			\hlinew{0.35pt}
			\multirow{2}{*}{Dataset} & \multicolumn{2}{c|}{M-LOSP} & \multicolumn{2}{c|}{HICODE} & \multicolumn{2}{c|}{IRLS-auto}  & \multicolumn{2}{c}{IRLS-gt} \\
			\cline{2-9}
			& \multicolumn{1}{c}{Each layer}  & \multicolumn{1}{c|}{Mean}  & \multicolumn{1}{c}{Each layer}  & \multicolumn{1}{c|}{Mean} & \multicolumn{1}{c}{Each layer}  & \multicolumn{1}{c|}{Mean} & \multicolumn{1}{c}{Each layer}  & \multicolumn{1}{c}{Mean} \\
			% \midrule
			\hline
			SynL2\_1	 & 0.760; 0.523  & 0.642  & 0.803; 0.806  & 0.804  & 0.918; 0.817  & \textbf{0.868}  & 0.947; 0.894  & 0.921  \\
			SynL2\_2	 & 0.832; 0.388  & 0.610  & 0.847; 0.910  & 0.878  & 0.991; 0.972  & \textbf{0.981}  & 0.999; 0.966  & 0.983  \\
			SynL2\_3	& 0.896; 0.387  & 0.642  & 0.841; 0.795  & 0.818  & 0.992; 0.898  & \textbf{0.945}   & 0.998; 0.941  & 0.970  \\
			SynL3\_1	 & 0.667; 0.352; 0.159  & 0.393  & 0.766; 0.681; 0.530  & 0.659  & 0.868; 0.760; 0.510  & \textbf{0.713}  & 0.853; 0.767; 0.603  & 0.741  \\
			SynL3\_2	 & 0.706; 0.305; 0.077  & 0.363  & 0.749; 0.776; 0.679  & 0.734	  & 0.984; 0.946; 0.622  & \textbf{0.851}   & 0.995; 0.934; 0.619  & 0.849  \\
			SynL3\_3	 & 0.721; 0.423; 0.059  & 0.190  & 0.832; 0.770; 0.609  & 0.737  & 0.965; 0.938; 0.730  & \textbf{0.878}	 & 0.955; 0.949; 0.713  & 0.872  \\
			% \bottomrule
			\hlinew{0.7pt}
	\end{tabular}}}
% 	\begin{tablenotes}
% 		\footnotesize
		%\item * The values in ``each layer'' are detection results of layer 1, layer 2 and layer 3 (if available) in order.
% 	\end{tablenotes}
\end{table*}

\begin{figure}
	\centering
	\includegraphics[width=0.8\linewidth]{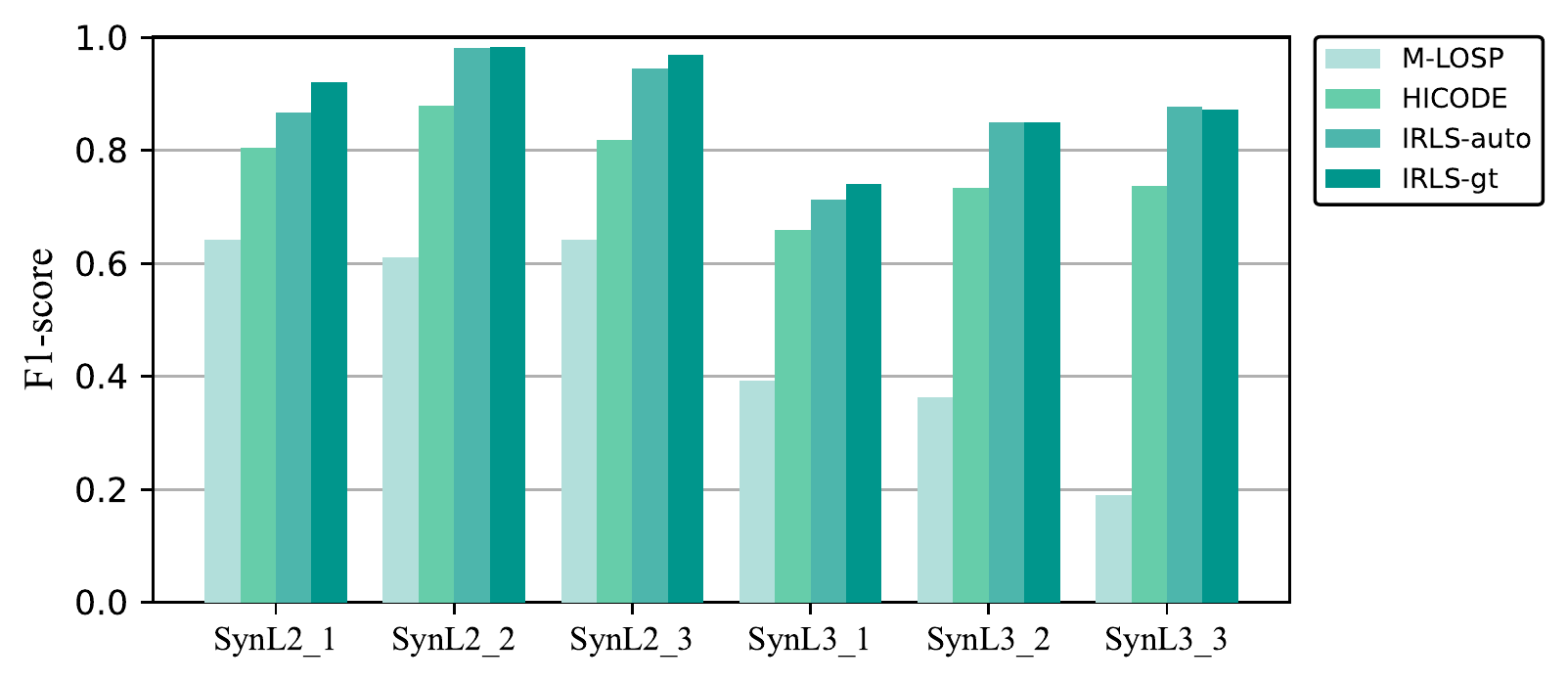}
	\caption{Comparison of the mean F1-score on synthetic datasets.}
	\label{fig:syn}
\end{figure}

The F1-scores between the ground truth communities and the detected communities by the four methods are shown in Table \ref{tab:SYNresults}, including the individual scores of each layer and the mean values.
Since the IRLS-gt can only be performed on the ideal condition, we mark the best results in bold among the other three methods that determine the community size automatically. We also plot histograms of the mean F1-scores in Figure \ref{fig:syn}. IRLS-gt and IRLS-auto both have good accuracy when conducted on the synthetic datasets, outperforming the two baselines. Due to the comparatively simple network structures of the two layers, the F1-scores on SynL2\_1, SynL2\_2 and SynL3\_3 are higher than that on the other three-layer datasets. And the detection results on SynL2\_1 and SynL3\_1 illustrate that the communities in networks approximated by the power law distribution are more difficult to be mined.
We observe that IRLS-gt outperforms the two baselines significantly on the detection accuracy. 

\subsection{Comparison on Real-world Datasets}
Rice\_Grad and Rice\_Ugrad are two social network datasets presented in~\cite{mislove2010you}, containing two sets of Facebook users of graduate and undergraduate students from Rice university and their connections, collected on May 17th, 2008. Each student has attributes of matriculation year, residential college and department, \etc. When a partition divided by an attribute is of comparatively high modularity, it can be considered as a layer of the network, where the connected components of the students with the same values form the communities. We extract four more networks from Facebook100 datasets~\cite{traud2012social}: MSU, FSU, Maryland and Texas, and use two attributes jointly to determine the layers and communities due to the larger scale of these networks. Students with chosen attributes missed are removed from the networks. The detailed statistics of the real-world datasets are shown in Table \ref{tab:realworld}.

\begin{table*}
	\scriptsize
	\renewcommand{\arraystretch}{1.3}
	\caption{Statistics on real-world datasets.}
	\label{tab:realworld}
	\centering
	\resizebox{\textwidth}{15mm}{
	\begin{threeparttable}
		\begin{tabular}{l l l l l l l l}
			\toprule
			Dataset  & $n$    & $m$   & $n_L$ & Chosen attributes & Modularity   & \#Communities & $\overline{|C|}$ \\
			\midrule
			Rice\_Grad      & 503 & 3,256 & 3 & department; school; year & 0.588; 0.584; 0.200   & 28; 9; 15  & 16.86; 53.78; 29.60\\
			Rice\_Ugrad     & 1,220   & 43,208  & 2 & college; year            & 0.385; 0.259       & 10; 7         & 120.10; 171.00\\
			MSU       & 13,294 & 289,510 & 2 & year+high school; year+dorm        & 0.147; 0.126       & 1396; 749      & 7.23; 14.28\\
			FSU     & 11,570 & 304,778 & 2 & year+dorm; year+high school        & 0.102; 0.081       & 870; 1325      & 10.13; 6.20\\
			Maryland & 10,404 & 303,889 & 2 & year+dorm; year+high school       & 0.135; 0.087       & 464; 1266       & 19.91; 6.03\\
			Texas    & 10,327 & 254,763 & 2 & year+dorm; year+high school       & 0.140; 0.086       & 585; 1447       & 14.76; 5.25\\
			\bottomrule
		\end{tabular}
		
		\begin{tablenotes}
			\footnotesize
			\item * Communities containing only one node are ignored.
		\end{tablenotes}
	\end{threeparttable}}
\end{table*}

For real-world datasets, IRLS-gt yields the highest F1-scores on all the five datasets. If the actual sizes of the communities are unknown, the detection accuracy of IRLS-auto is a little affected, but still superior to the two baselines. And it even performs better than IRLS-gt on MSU. Not considering the varying cohesiveness of different layers, M-LOSP gives the lowest performance among the four methods. HICODE significantly outperforms the multiple local community detection algorithm, indicating the reduction method is necessary for the hidden community detection.
Yet HICODE is not comparable to our IRLS methods, which are specially designed for the local hidden community detection problem. Note that there are only hundreds or thousands of nodes in Rice\_Grad and Rice\_Ugrad, less than the number of sampling threshold. That is to say, HICODE actually works on the whole networks of the datasets, but it is still not able to attain the best detection results, strongly demonstrating the necessity of our method.

\makeatletter
\def\hlinew#1{%
	\noalign{\ifnum0=`}\fi\hrule \@height #1 \futurelet
	\reserved@a\@xhline}
\makeatother

\begin{table*}
	\scriptsize
	\renewcommand{\arraystretch}{1.3}
	\caption{Comparison of F1-score on real-world datasets.}
	\label{tab:RWresults}
	\centering
	\resizebox{\textwidth}{15mm}{
	\scalebox{1}{
		\begin{tabular}{l| l l| l l| l l| l l}
			\hline
			% \toprule
			\hlinew{0.35pt}
			\multirow{2}{*}{Dataset} & \multicolumn{2}{c|}{M-LOSP} & \multicolumn{2}{c|}{HICODE} & \multicolumn{2}{c|}{IRLS-auto}  & \multicolumn{2}{c}{IRLS-gt} \\
			\cline{2-9}
			& \multicolumn{1}{c}{Each layer}  & \multicolumn{1}{c|}{Mean}  & \multicolumn{1}{c}{Each layer}  & \multicolumn{1}{c|}{mean} & \multicolumn{1}{c}{Each layer}  & \multicolumn{1}{c|}{Mean} & \multicolumn{1}{c}{Each layer}  & \multicolumn{1}{c}{Mean} \\
			% \midrule
			\hline
			Rice\_Grad   & 0.543; 0.041; 0.018  & 0.201  & 0.537; 0.263; 0.100  & 0.300   & 0.502; 0.295; 0.139  & \textbf{0.312}  & 0.593; 0.461; 0.311  & 0.455\\
			Rice\_Ugrad  & 0.033; 0.412  & 0.223  & 0.547; 0.082  & 0.315  & 0.477; 0.227  & \textbf{0.352}  & 0.545; 0.372  & 0.458  \\
			MSU  & 0.236; 0.048  & 0.142  & 0.230; 0.214  & 0.222  & 0.460; 0.119  & \textbf{0.290}	 & 0.382; 0.124  & 0.253  
			\\
			FSU  & 0.051; 0.143  & 0.097  & 0.083; 0.130  & 0.106  & 0.119; 0.205  & \textbf{0.162}  & 0.154; 0.182  & 0.168  
			\\
			Maryland  & 0.084; 0.162  & 0.123  & 0.075; 0.220  & 0.148  & 0.121; 0.240  & \textbf{0.181}  & 0.127; 0.261  & 0.194  \\
			Texas  & 0.076; 0.160  & 0.118  & 0.085; 0.096  & 0.091  & 0.086; 0.168  & \textbf{0.127}	 & 0.153; 0.190  & 0.171  \\
			% \bottomrule
			\hlinew{0.7pt}
	\end{tabular}}}

\end{table*}

\begin{figure}
	\centering
	\includegraphics[width=0.8\linewidth]{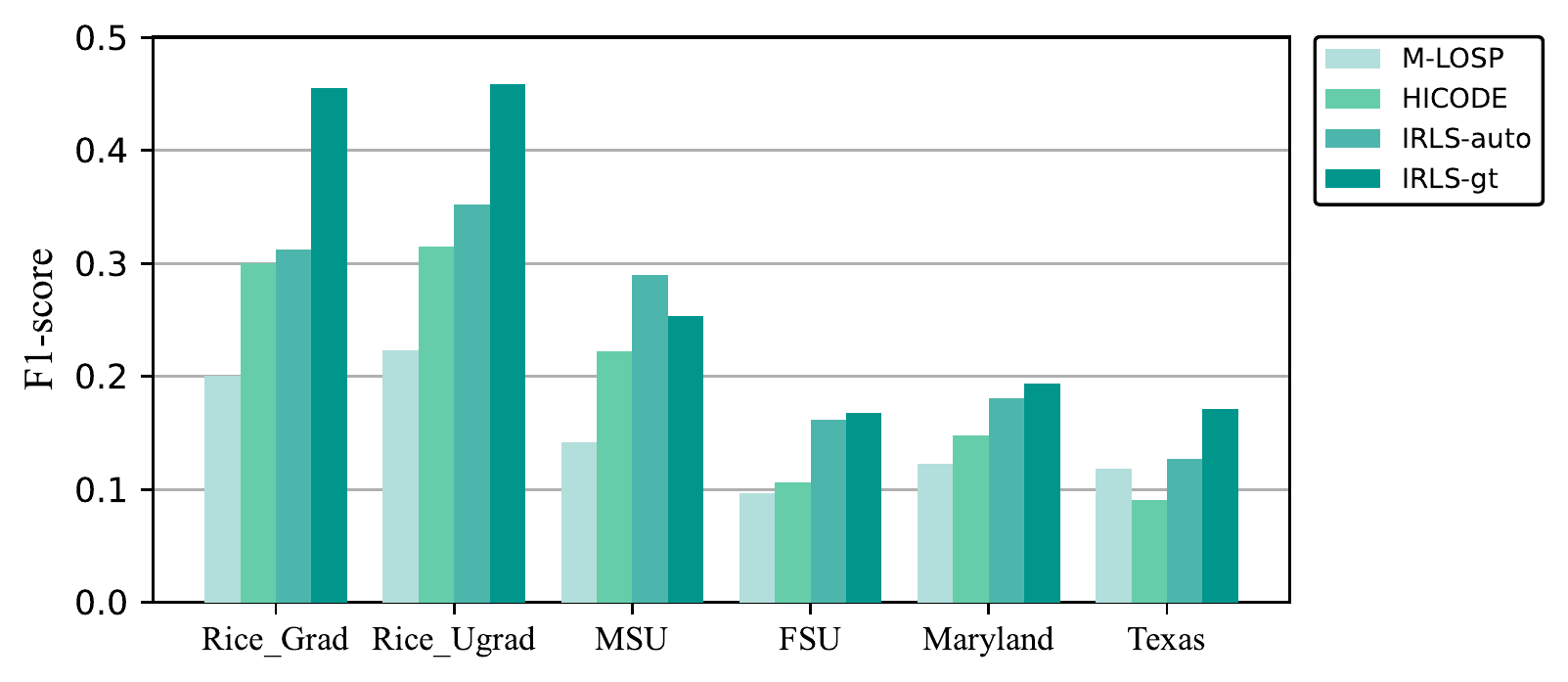}
	\caption{Comparison of the mean F1-score on real-world datasets.}
	\label{fig:real}
\end{figure}

\subsection{Evaluation on IRLS} \label{evalution}
In this subsection, we evaluate the number of iterations, $T$, and the maximal seed set size, $n_{\text{set}}$, to determine the appropriate values. IRLS uses an iterative refinement process to improve the detection accuracy. More iterations can lead to better results, while they bring higher calculation overhead in the meantime. We apply 20-iteration IRLS-gt and IRLS-auto on SynL3\_3, and the convergence line charts are shown in Figure \ref{fig:convgt} and Figure \ref{fig:convauto}. The two types of IRLS show similar convergence trends. According to the figures, the F1-scores of three layers increase rapidly at the first few iterations. Layer 1 and layer 2's detection results keep basically stable since the $10^{th}$ iteration, while layer 3's scores grow a little afterwards, which makes a tiny influence on the mean values. The experiments show similar or faster convergence on the other datasets. In order to obtain good detection accuracy without consuming too much time, we choose $T=10$ for all datasets.

\begin{figure}[htbp]
\centering
\begin{minipage}[t]{0.49\textwidth}
\centering
\includegraphics[width=6.5cm]{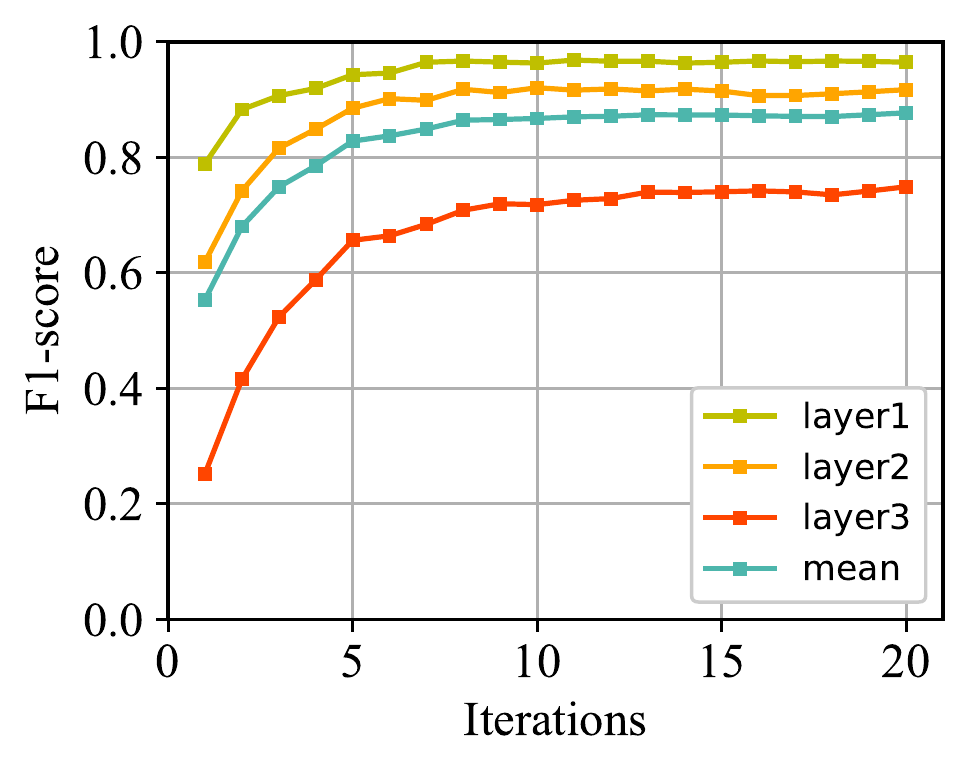}
\caption{Evaluation on the number of iterations $T$ of IRLS-gt.}
\label{fig:convgt}
\end{minipage}
\begin{minipage}[t]{0.49\textwidth}
\centering
\includegraphics[width=6.5cm]{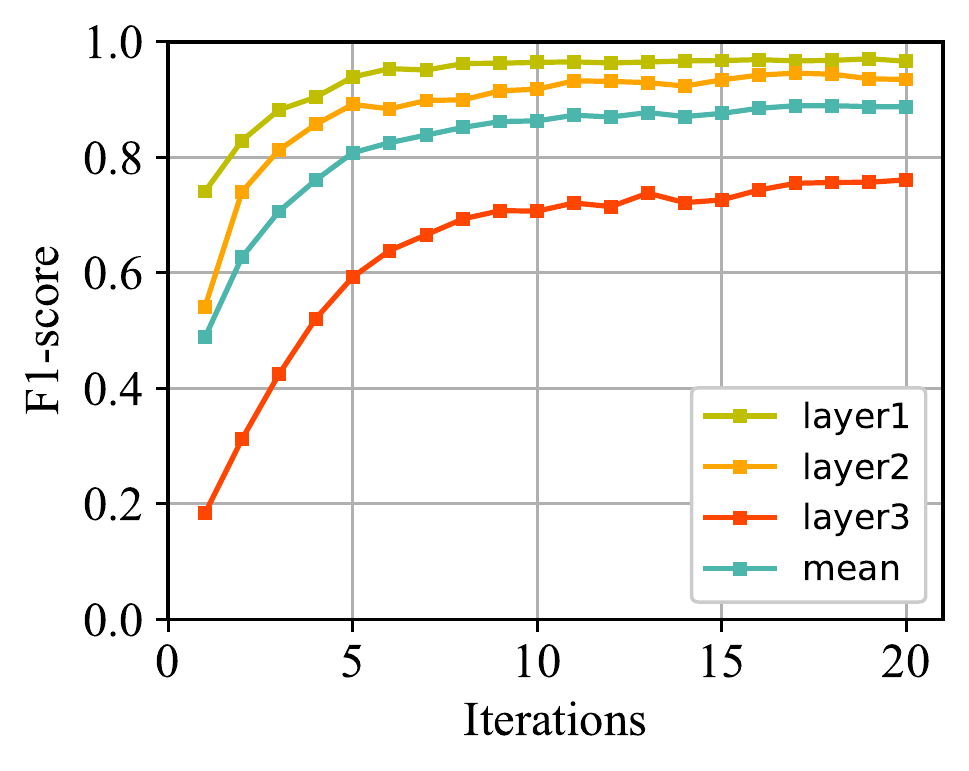}
\caption{Evaluation on the number of iterations $T$ of IRLS-auto.}
\label{fig:convauto}
\end{minipage}
\end{figure}

We also test two types of IRLS on SynL3\_3 with different maximal seed set sizes during the seed set augmentation stage, recording the corresponding F1-scores and running times (for each experiment, we calculate the average running time of each iteration), which are shown in Figure \ref{fig:seedF1} and Figure \ref{fig:seedTime}. On the one hand, a seed set with a tiny size may not have enough ability to find out the whole community. On the other hand, if the seed set is too large to include some nodes from other communities, it will also interfere with the detection accuracy. As we can see from Figure \ref{fig:seedF1}, two lines of F1-scores keep increasing at first and reach their peak values when the seed set size is 18. Benefiting from our efficient augmentation strategy, the seed set usually takes only few steps from a single node to the maximal size, and it will not cause much more time consumption when a larger size is chosen. We can draw the same conclusion from Figure \ref{fig:seedTime}. The running times of different seed set sizes do not make much difference. Taking both aspects into consideration, we set the maximal seed set size to 18 when conducting experiments on synthetic datasets. For real-world datasets, as there exist a number of communities of small sizes and they cannot be detected correctly when using a seed set containing more than ten nodes, we choose 9 as the maximal seed set size for these datasets.
% For brevity, we do not show these comparison results.

\begin{figure}[htbp]
\centering
\begin{minipage}[t]{0.49\textwidth}
\centering
\includegraphics[width=6.5cm]{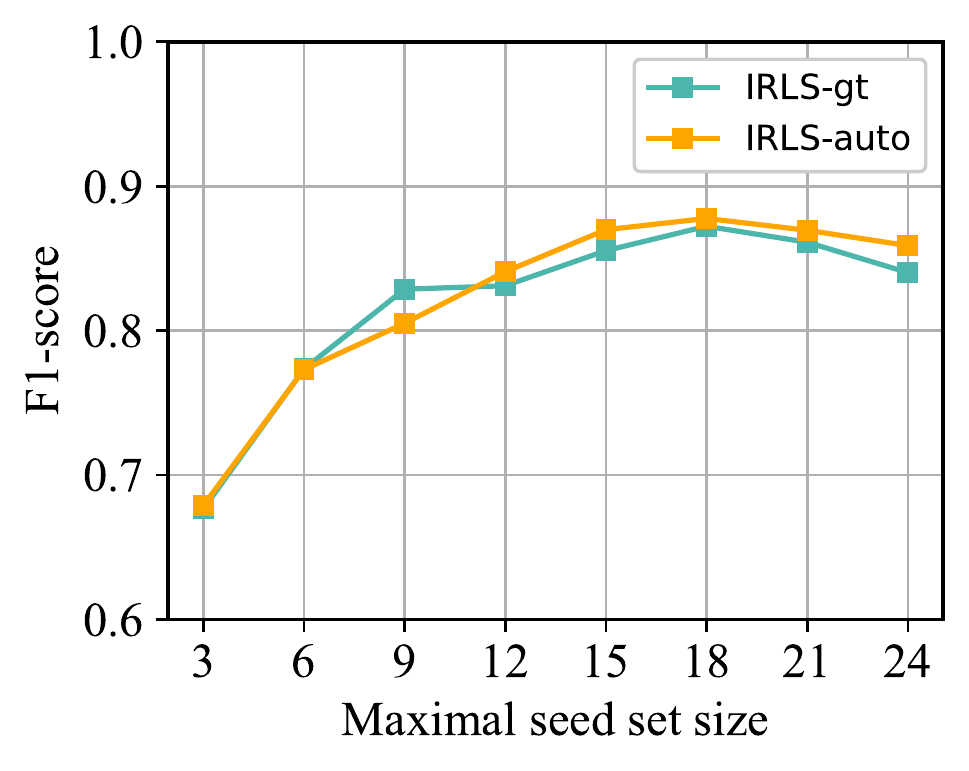}
\caption{Evaluation of the maximal seed set size on F1-score.}
\label{fig:seedF1}
\end{minipage}
\begin{minipage}[t]{0.49\textwidth}
\centering
\includegraphics[width=6.5cm]{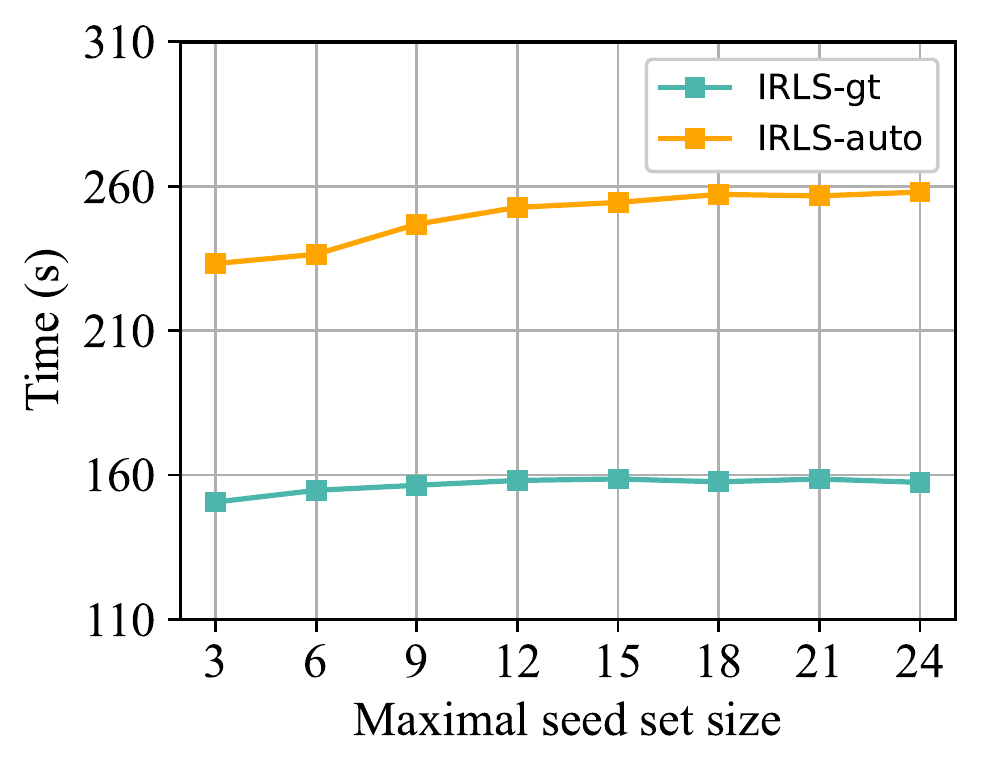}
\caption{Evaluation of the maximal seed set size on running time.}
\label{fig:seedTime}
\end{minipage}
\end{figure}

\section{Conclusion}

In this paper, we introduce a new problem called the local hidden community detection, which aims to find all the local communities containing the query seed node in social networks with multiple layers of different cohesiveness. We design a new algorithm called the Iterative Reduction Local Spectral (IRLS) method to address this problem. IRLS adds an adaptive seed set augmentation process to the local community detection, and uses the weighted local modularity to determine the boundary of a local community when its actual size is unknown. We apply the method on the sampled subgraph to uncover the local community. After that, nodes in the detected community and all edges connecting with them are temporarily removed from the network and we seek for a partition in the remaining subgraph. Then all the communities obtained are weakened to break the structure of the current layer and reveal the hidden layers. We repeat these steps until all the local communities in different layers are found and boosted. Both theoretical analysis and experimental results demonstrate the necessity and superiority of our method. Two types of IRLS methods are tested on synthetic datasets and real-world networks, showing great performance and surpassing the related baselines designed for multiple local community detection or global hidden community detection.

Local hidden community detection is a new research topic with more potential questions to explore. For example, how to deal with the situation that different nodes in the network belong to different numbers of communities? And how to distinguish layers with similar degrees of cohesiveness?
We will continue to follow this direction and endeavor to present more investigation in our future work. 

% \printcredits
\section*{Acknowledgement}
This work is supported by National Natural Science Foundation of China (62076105).

\bibliographystyle{ACM-Reference-Format}
\bibliography{main.bbl}

%\blindtext[20]
\end{document}